\documentclass[12pt,a4paper]{article}
\usepackage{amsmath,amssymb}
\usepackage{amsthm}
\usepackage{relsize}
\usepackage[final]{pdfpages}
\usepackage{amsfonts}
\usepackage[ansinew]{inputenc}
\usepackage{bookman}
\usepackage{pstricks}
\usepackage{setspace}
\newtheorem{definition}{Definition}
\newtheorem{example}{Example}
\newtheorem{theorem}{Theorem}
\newtheorem{lemma}{Lemma}
\newtheorem{corollary}{Corollary}
\DeclareMathOperator{\ef}{\mathbb{F}}
\DeclareMathOperator{\aaa}{\boldsymbol{A}}

\DeclareMathOperator{\gt}{\boldsymbol{GT}}
\DeclareMathOperator{\enn}{\mathbb{N}}
\DeclareMathOperator{\supp}{Supp}
\DeclareMathOperator{\ddd}{\boldsymbol{D}}
\DeclareMathOperator{\kr}{Ker}
\DeclareMathOperator{\lep}{LE}

\DeclareMathOperator{\wdeg}{wdeg}
\usepackage[top=3cm, bottom=3cm, left=2.5cm,headheight=2cm,headsep=1cm, right=2.5cm]{geometry}
\begin{document}

\title{A Dynamical System-based Key Equation for Decoding One-Point Algebraic-Geometry Codes}
\author{Ramamonjy ANDRIAMIFIDISOA\\
\small ramamonjy.andriamifidisoa@univ-antananarivo.mg \and
        Rufine Marius LALASOA \\
\small        larissamarius.lm@gmail.com\and
        Toussaint Joseph RABEHERIMANANA\\
\small     rabeherimanana.toussaint@yahoo.fr
       }
\date{\today}
\maketitle

\begin{abstract}
   A closer  look at linear recurring sequences allowed us to define the multiplication  of a univariate polynomial and a sequence, viewed as a power series with another variable, resulting in another sequence.  Extending this operation,  one gets the multiplication of matrices of multivariate polynomials and vectors of powers series. A dynamical system, according to U. Oberst is then the kernel of the linear mapping of modules defined by a polynomial matrix by this operation. Applying these tools in the decoding of the so-called one point algebraic-geometry codes, after showing that the syndrome array, which is the general transform of the error in a received word is a linear recurring sequence, we  construct a dynamical system. We then prove that this array  is the solution of  Cauchy's homogeneous equations with respect to the dynamical system. The aim of the Berlekamp-Massey-Sakata Algorithm in the decoding process being the determination of  the syndrome array, we have proved that in fact, this algorithm solves the Cauchy's homogeneous equations with respect to a dynamical system.
\end{abstract}

\tableofcontents

\section{Introduction}
\label{intro}
S. Sakata,   in \cite{Sakata-0,Sakata-2}, generalized the famous Berlekamp-Massey algorithm (\cite{Massey}) to the  two and multidimensional case. The result was the  (again) famous  Berlekamp-Massey-Sakata  (BMS) algorithm, whose aim is  to find a Gr\"{o}bner basis of the set of characteristic polynomials of a multidimensional sequence. He also used the algorithm to decode \textit{algebraic-geometry} (AG) codes (\cite{Sakata-3,Sakata-4}). The main difficulties  is that Sakata's  papers  involves  many  difficult notations and calculations.

Heegard and Saints (\cite{HS}) gave a shorter and clearer version of this algorithm,  and explained that, in the framework of the decoding process, the algorithm computes a sufficiently number of terms of the syndrome array and construct  sets of polynomials which ``converges'' to a Gr\"{o}bner basis, which allows the calculation of the syndrome array.

Since then, the BMS algorithm  has been   refined and widely used   by many authors, see \cite{Bras,Cox-2,Faugere,Kuijper,Sullivan},  and also D.~Augot, ``\textit{Les codes alg\'{e}briques principaux et leur d\'{e}codage}", Journ\'{e}es nationales du calcul formel. Luminy, mai 2010 and J.~Berto\-mieux and J. C.~Faug\`{e}re, ``\textit{In-depth comparison of the Berlekamp-Massey-Sakata and the Scalar-FGLM algorithms: the non adaptive variants}",\\ arXiv:1709.07168 [cs.SC] (2017).

Therefore, due to its importance,  we present here a  new explanation of the BMS algorithm, in the framework of the decoding process of  \textit{one point algebraic-geometry codes},  as in \cite{HS}.  To construct these codes,  one starts from a \textit{smooth irreducible projective curve} which have a unique point only at the \textit{hyperplane at infinity}, and a finite set of points of the curve, distinct from the point at infinity.\  The code is the defined as evaluations  of   certain \textit{rational functions} (\eqref{cr} and \eqref{rf}) on the curve on the set  of points (\eqref{code}). The conditions these functions have to satisfy is that they  have a unique \textit{pole}, which is  the point at the infinity, and\ moreover, the order of this pole is  less than an appropriate number, which satisfies an inequality involving the \textit{genus} of the curve and the number of evaluation points (\eqref{a}).

An important tool we use is the \textit{general transform} (Definition \eqref{gt}).  The crucial starting point of our result is that the general transform of the error in a received word is a \textit{linear recurring sequence} (Corollary \ref{E-LRS}). Here is where the notion of dynamical system can be introduced : the \textit{orthogonal} of the syndrome array is a polynomial module, and therefore has a \textit{Gr\"{o}bner basis}. We consider the dynamical system defined by this basis.\

We   prove in our main theorem (Theorem \ref{main})  that
the syndrome array of a received word is the solution of the  \textit{Cauchy's homogeneous problem }(Definition \ref{Cauchy-h}) with respect to the above dynamical system,  under the input/output representation (\eqref{io} and \eqref{Code-Sys}), with an appropriate initial data, defined on a ``Delta-set"(\eqref{Delta-set}).

 Our theorem provides  a new equation for the decoding problem.
We hope that our equation is a good starting point for understanding the BMS algorithm and decoding one point AG codes because it provides a clean and elegant  algebraic presentation of the algorithm and the decoding problem.

This paper is organized as follows: in section \ref{ds-cp}, we introduce  Oberst's dynamical systems theory and the Cauchy's homogeneous problem. In section \ref{ag-codes}, we present results about projective curves and one-point algebraic-geometry codes. In the last section  \ref{ce}, we state and prove our main theorem.

As we already mentioned in the abstract, the simple notion of \textit{linear recurring sequence } is useful to understand the operation denoted by  ``$\circ$" in  Section \ref{ds-cp}. A sequence $a=(a_n)_{n\in\enn}$ of elements of a commutative field  $\ef$ is said to be a linear recurring sequence (LRS) if the following equality holds:
\begin{equation}\label{LRS}
   P_0 a_n+P_1a_{n+1}\cdots +P_ia_{n+i}+\cdots+P_Na_{n+N}=0\quad\text{for}\quad n\in \enn,
 \end{equation}
where $N\geqslant 1$ is an integer, $P_i\in\ef$ for $i=0,\ldots,N$ with $P_N\neq 0$. Using equation \eqref{LRS}, we have that
\begin{equation*}
  a_{n+N} = -\frac{1}{P_N}(P_0 a_n+P_1a_{n+1}\cdots +P_ia_{n+i}+\cdots+P_Na_{n+N-1}),
\end{equation*}
so that we can calculate $a_{n+N}$ using the $N$ previous terms of the sequence, which are $a_n,\ldots,a_{n+N-1}$. \\
We observe that the left hand side of \eqref{LRS} is the $n$-th term of a new sequence  of elements of $\ef$. Denoting this sequence by $b=(b_n)_{n\in\enn}$, we have
\begin{equation}\label{b}
  b_n = \sum_{i=0}^NP_ia_{n+i}\quad\text{for}\quad n\in\enn.
\end{equation}
 Now, construct the univariate polynomial
\begin{equation*}
  P(X) =\sum_{i=0}^N P_i X^i\in \ef[X],
\end{equation*}
and write the sequences $a$ and $b$ as  power series in another variable, say $Y$:
\begin{equation*}
  a=a(Y) = \sum_{n=0}^\infty a_n Y^n,\quad b=b(Y)=\sum_{n=0}^\infty b_n Y^n.
\end{equation*}
We say that $b(Y)$ is the \textit{product} of $P(X)$ and $a(Y)$ and write
\begin{equation*}
b(Y)= P(X)\circ a(Y).
\end{equation*}
Using \eqref{b}, we have
\begin{equation}\label{rond}
P(X)\circ a(Y) =\sum_{n=0}^\infty\Bigl(\sum_{i=0}^NP_ia_{n+i}\Bigr)Y^n
\end{equation}
(compare with \eqref{shift}). The polynomial $P(X)$ is called a \textit{characteristic polynomial} of the sequence $a$.

\section{Oberst's algebraic dynamical systems and the Cauchy's homogeneous problem}\label{ds-cp}
Let  $\ef$ be a commutative field. For an integer $r\geqslant 1$, let  $X_1,\ldots,X_r$ and $Y_1,\ldots,Y_r$ distinct \textit{variables}.  The letter  $X$ (resp. $Y$) will denote the set of variables $(X_1,\ldots,X_r) $ (resp. $(Y_1,\ldots,Y_r$)).  For $ \alpha = (\alpha_1,\ldots,\alpha_r)\in\enn^r$, we define $X^\alpha$ (resp. $Y^\alpha$)  as the product
\begin{equation*}\label{X-Y-alpha}
  X^\alpha=X_1^{\alpha_1}\cdots X_r^{\alpha_r} \;({\rm resp.}\; Y^\alpha = Y_1^{\alpha_1}\cdots Y_r^{\alpha_r}) .
\end{equation*}
Let $\ddd=\ef[X_1,\ldots ,X_r]=\ef[X]$  be the $\ef$-vector space of the polynomials with the $r$ variables $X_1,\ldots, X_r$ and entries in $\ef$. An element of $\ddd$ can be  uniquely written as
\begin{equation*}
d(X_1,\ldots,X_r) =d(X)=\sum_{\alpha\in\enn^r} d_\alpha X^\alpha\quad  \text{ with}\quad  d_\alpha\in\ef\quad  \text{for all } \alpha\in\enn^r,
\end{equation*}
where $d_\alpha=0$ except for a finite number of $\alpha$'s. We fix a \textit{monomial ordering} $\leqslant_T$  on $\enn^r$, (\cite{Cox,Oberst}) which is then a \textit{well ordering}. For a non-zero element $d(X)\in\ddd$, we define the \textit{leading exponent} of $d(X)$ by
\begin{equation}\label{le}
\lep(d(X)) = \max_{\leqslant_T}\{\alpha\in\enn^r\;\vert\; d_\alpha\neq 0\}\in\enn^r.
\end{equation}
Let $\aaa=\ef[[Y_1,\ldots,Y_r]]=\ef[[Y]]$ be $\ef$-vector space of the formal power series with the variables $Y_1,\ldots,Y_r$ and entries in $\ef$. An element of $\aaa$ can be uniquely written as
\begin{equation*}
  W(Y_1,\ldots,Y_r) = W(Y) =\sum_{\alpha\in\enn^r}W_\alpha Y^\alpha,
\end{equation*}
where $W_\alpha\in\ef$ for all $\alpha\in\enn^r$.\\

  For integers $k,l\geqslant 1$, the set of matrices with $k$ rows and $l$ columns with entries in $\ddd$ is denoted by $\ddd^{k,l}$. An element $R(X)\in\ddd^{k,l}$   is of the form
\begin{equation*}
R(X) = (R_{ij}(X))_{1\leqslant i\leqslant k, 1\leqslant j\leqslant\ l},
\end{equation*}
where $R_{ij}(X)\in\ddd$ for $i=1,\dots, k$ and $j=1,\ldots, l$. With the multiplication by polynomials as external operation of $\ddd$ on $\ddd^{k,l}$, this latter becomes $\ddd$-module. The notation $\ddd^l$ (resp.  $\aaa^l$) will be for the set of polynomials   with one row and $l$ columns (resp.  power series in $\aaa$ with $l$ rows and one column).\\

The external operation, (also called \textit{multiplication}) of  $\ddd$ on $\aaa$ is defined by
\begin{align}\label{shift}
\begin{split}
\ddd\times\aaa&\longrightarrow\aaa\\
(d(X),W(Y))&\longmapsto d(X)\circ W(Y) = \sum_{\alpha\in\enn^r}(\sum_{\beta\in\enn^r}d_\beta W_{\alpha+\beta})Y^\alpha.
\end{split}
\end{align}
This operation provides $\aaa$ with a $\ddd$-module structure.  The set $\aaa^l$ becomes a $\ddd$-module too, with the external operation
\begin{align}\label{m-shift}
\begin{split}
\ddd\times\aaa^l&\longrightarrow\aaa^l\\
(d(X),(W_j(Y))_{j=1,\ldots,l})&\longmapsto (d(X)\circ W_j(Y))_{j=1,\ldots,l}.
\end{split}
\end{align}
More generally, given $R(X)\in\ddd^{k,l}$, the following mapping, also denoted by $R(X)$,  is a $\ddd$-linear mapping of modules
\begin{align}
\begin{split}
R(X) : \aaa^l&\longrightarrow\aaa^k\\
W(Y) &\longmapsto R(X)\circ W(Y)
\end{split}
\end{align}
where
\begin{align}\label{rw}
\begin{split}
R(X)\circ W(Y) &=\left(
           \begin{array}{c}
             \sum_{j=1}^lR_{1j}(X)\circ W_j(Y) \\
             \vdots \\
             \sum_j^lR_{kj}(X)\circ W_j(Y) \\
           \end{array}
         \right)\\
         & = \left(
     \begin{array}{c}
       \sum_{\rho\in\enn^r}(\sum_{j=1}^l\sum_{\alpha\in\enn^r} R_{1j\alpha}W_{j(\alpha+\rho)})Y^\rho \\
       \vdots \\
       \sum_{\rho\in\enn^r}(\sum_{j=1}^l\sum_{\alpha\in\enn^r} R_{ij\alpha}W_{j(\alpha+\rho)})Y^\rho \\
       \vdots \\
       \sum_{\rho\in\enn^r}(s\sum_{j=1}^l\sum_{\alpha\in\enn^r} R_{kj\alpha}W_{j(\alpha+\rho)})Y^\rho \\
     \end{array}
   \right)
\end{split}
\end{align}
is $\ddd$-linear (\cite{Andr-1,Oberst}. Note that this expression of  $R(X)\circ W(Y)$  is similar to that of the usual matrix-vector multiplication). Its kernel is then a $\ddd$-submodule of $\aaa^l$. This legitimates the following definition:
\begin{definition}[Oberst, \cite{Oberst}]\label{sys}An algebraic dynamical system (or simply a \textit{system}) is a $\ddd$-submodule of $\aaa^l$ of the form
\begin{equation*}
  S = \ker R(X) =\{W(Y)\in\aaa^l\;\vert\; R(X)\circ W(Y) = 0\}
\end{equation*}
where $R(X)\in\ddd^{k,l}$ and also denotes  the $\ddd$-linear mapping of $\ddd$-modules defined by \eqref{rw}.
\end{definition}

The integer $r$ is the \textit{dimension} of the system.  Willems treated the one-dimensional case only. An element $W$ of a system $S$ is called a \textit{trajectory}.
\begin{example}[Linear recurring sequence]\label{LRS}Take $r=1$. Then $\ef[X]$ is the set of univariate polynomials in $X$ and $\ef[[Y]]$ the set of power series in the unique variable $Y$. A polynomial $P(X)\in\ef[X]$ defines the dynamical system
\begin{equation*}
  \ker P(X) = \{a(Y) = \sum_{n=0}^\infty a_nY^n\in\ef[|Y]]\;|\;P(X)\circ W(Y)=0\}.
\end{equation*}
  If $P(X)=0$, then $\ker P(X)=\ef[[Y]]$, otherwise,  using \eqref{shift}, for $r=1$, we are in the situation in \eqref{rond}, so that the elements of $\kr P(X)$ are the linear recurring sequences having $P(X)$ as a characteristic polynomial.
 \end{example}

For a subset $P$ of $\ddd^l$  and a subset $Q$ of $\aaa^l$, we define their \textit{orthogonals} by
\begin{align*}
&P^\perp = \{W(Y)\in\aaa^l\;|\;d(X)\circ W(Y)=0\;\text{for}\;d(X)\in P\}\subset\aaa^l\\
&Q^\perp = \{ d(X)\in\ddd^l\;|\;d(X)\circ W(Y)=0\; \text{for}\;W(Y)\in\aaa^l\}\subset\ddd^l\label{syst-orth}.
\end{align*}
$P^\perp$ is a $\ddd$-submodule of $\aaa^l$ and $Q^\perp$ is a $\ddd$-submodule of $\ddd^l$ (\cite{Oberst}).
\begin{example}\label{orths}For a non-zero polynomial $P(X)\in\ddd$, the set $P(X)^\perp=\{P(X)\}^\perp$ is those of the LRS  having $P(X)$ as a characteristic polynomial. For a power series  $W(Y)\in \aaa$, the set $W(Y)^\perp=\{W(Y)\}^\perp$ is those  of the characteristic polynomials of $W(Y)$ and the zero polynomial.

\end{example}

In \cite{Oberst}, it is proven that every system $S$ admits an Input/Output  representation
\begin{equation}\label{io}
S = \Biggl\{\left (\begin {array}{c} U\\\noalign{\medskip}V\end
{array}\right ) \in\aaa^m\times\aaa^p \;\vert\;P(X)\circ
V = Q(X)\circ U \Biggr\},
\end{equation}
where $m,p\geqslant 1$ are integers with
\begin{equation}
 l = m\;+\;p,\quad P\in\ddd^{k,p},\quad Q\in\ddd^{k,m},
\end{equation}
the columns of $P$ being $\mathbf{K}$-linearly independent with $\mathbf{K}=\ef(X_1,\ldots,X_r)$ and
\begin{equation}
  {\rm rank}( P) = {\rm rank}(R) = p.
\end{equation}
The system written in the form \eqref{io} is called an \textit{I/O system}.

Now, we need some notations for an integer $p\geqslant 0$, we write
\begin{equation}\label{p}
[p]=\{1,\ldots,p\},
\end{equation}
and $\Gamma$ denotes a subset of $[p]\times\enn^r$
(If $p=1$, then we identify $[p]\times\enn^r$  with $\enn^r$). We may identify $\;\;\ef^{[p]\times\enn^r}$ with  $\aaa^p$ and consider $\ef_p^\Gamma$ as a subset of $\aaa^p$, where $\ef^\Gamma$ is the set of mappings from $\Gamma$ to $\ef$.\\

\begin{definition}[Oberst, \cite{Oberst}]\label{Cauchy-h}The homogeneous  Cauchy
problem  $(P(X),0,\Gamma)$ for the I/O system  \eqref{io} is the system of equations
\begin{equation}\label{chom}
  \begin{cases}
  \quad P(X)\circ V =  0 , \\
  \quad V_{\vert\Gamma}  =  V_0,\; V_0\in\ef^\Gamma,
  \end{cases}
\end{equation}
where the unknown is  $V\in\aaa^p$, the initial data being $V_0\in\ef^\Gamma$.
\end{definition}

\section{On point algebraic-geometry codes}\label{ag-codes}
 For algebraic geometry, we refer to \cite{Cox,Fulton} and the construction of one point AG codes, we refer to \cite{HS}. We recall here the basic notations and ideas  for the construction of such codes.\ From now on, $\ef_q$ denotes the Galois field with $q$ elements, where $q$ is a power of a positive prime integer. Let $\ef$ be the algebraic closure of $\ef_q$ and $r\geqslant 1$ and integer.

We write $X=(X_1,\ldots,X_r)$ as in section \ref{ds-cp}. We will use  the polynomial rings $\ef_q[x_1,\ldots,X_r], \ef[X_1,\ldots,X_r]$ and $\ef[X_0,\ldots,X_r]$, where $X_0$ is   another variable. We denote by $\mathbb{P}^r(\ef)$  the $r$-dimensional \textit{projective space} over $\ef$. An element  of $\mathbb{P}^r(\ef)$   is of the form $P=(a_0:a_1:\ldots:a_r)$, where $a_i\in\ef$. The \textit{hyperplane at infinity} is the set $\mathbb{H}^r_\infty$  of the points of the form $(0:a_1:\ldots:a_r)\in\ef^{r+1}$. One may then write (up to an isomorphism) $\mathbb{P}^r(\ef)=\ef^r\cup\mathbb{H}^r_\infty$, and identify a point $P=(a_1:\ldots:a_r)\in\ef_q$ with the point $P=(1:a_1:\ldots:a_r)\in\mathbb{P}^r(\ef)$.

We will consider a \textit{smooth irreducible projective curve} $\mathcal{X}$ \textit{defined over}\\ $\ef_q[X_1,\ldots,X_r]$. It is an \textit{affine variety} of \textit{dimension} $1$, defined by
\begin{equation*}
  \mathcal{X} =\{P=(a_0,\ldots,a_r)\in \mathbb{P}^r(\ef)\;\vert\;F(P)=0\;\text{for}\; F\in\mathcal{F}\},
\end{equation*}
where $\mathcal{F}$ is a set of \textit{homogeneous polynomials} of $\ef[X_0,X_1,\ldots,X_r]$.
The ideal of $\mathcal{X}$ is
\begin{equation*}
  I(\mathcal{X}) =\{F\in\ef[X_0,X_1,\ldots,X_r]\;\vert\;F(P)=0\;\text{for}\;P\in\mathcal{X}\}.
\end{equation*}
The \textit{coordinate ring} of $\mathcal{X}$ is the ring
\begin{equation}\label{cr}
  \ef[\mathcal{X}] = \ef[X_0,X_1,\ldots,X_r]/I(\mathcal{X}).
\end{equation}
The $\ef[\mathcal{X}]$ is an integral domain and  its  field of fractions is called the \textit{field of rational functions }on $\mathcal{X}$ and denoted by $\ef(\mathcal{X})$. We may write
\begin{align}\label{rf}
\begin{split}
\ef(\mathcal{X}) = \{f(X_0,X_1,\ldots,X_r)/g(X_0,X_1,\ldots,X_r)\;\vert\;f,g\in\ef[X_0,X_1,\ldots,X_r]\\
\text{and}\;g(X_0,X_1,\ldots,X_r)\notin I(\mathcal{X})\}.
\end{split}
\end{align}
The curve $\mathcal{X}$ is constructed from a smooth irreducible \textit{affine curv}e $\mathcal{X}_{aff}$ \textit{\textit{defined over}} $\ef_q[X_1,\ldots,X_r]$, which is of the form
\begin{equation*}
  \mathcal{X}_{aff}=\{P=(a_1,\ldots,a_r)\in\ef^r\;\vert\;F(P)=0\;\text{for}\;P\in\mathcal{G}\},
\end{equation*}
where $\mathcal{G}$ is a set of polynomials in $\ef[X_1,\ldots,X_r]$. The ideal of $\mathcal{X}_{aff}$ is
\begin{equation*}
  I(\mathcal{X}_{aff})=\{F\in\ef[X_1,\ldots,X_r]\;\vert\;F(P)=0\;\text{for}\; P\in\mathcal{X}_{aff}\}.
\end{equation*}
The terminology ``$\mathcal{X}$ (or $\mathcal{X}_{aff}$) defined over $\ef_q[X_1,\ldots,X_r]$" means that the ideal $I(\mathcal{X}_{aff})$ is generated by polynomials in $\ef_q[X_1,\ldots,X_r]$.
  As in \eqref{cr} and \eqref{rf}, we define the coordinate ring (resp. the field of rational functions) of $\mathcal{X}_{aff}$ :
\begin{align*}
&\ef[\mathcal{X}_{aff}] = \ef[X]/I(\mathcal{X}_{aff}),\\
&\ef(\mathcal{X}_{aff}) = \{f(X)/g(X)\;\vert\;f,g\in\ef[X]\;\text{and}\;g(X)\notin I(\mathcal{X}_{aff})\}.
\end{align*}

   The field of rational functions $\ef(\mathcal{X})$ is \textit{birationally equivalent} to $\ef(\mathcal{X}_{aff})$, so we may use this latter only.  Moreover, the projective  curve  we consider will have a unique point $Q$ lying at  the hyperplane at infinity  and is \textit{in special position} with respect to $Q$. Let $a$ be an integer verifying
\begin{equation}\label{a}
2g-2<a<n,
\end{equation}
where $g$ is the \textit{genus} of $\mathcal{X}$.    Let $\mathcal{L}(aQ)$ be the set of the functions $\phi$ on $\ef(\mathcal{X}_{aff})$ which have a unique \textit{pole} at $Q$, of \textit{order} less than $a$.   \\

 Let  $\mathcal{P}=\{P_1,\ldots,P_n\}$ a set of points of $\mathcal{X}$.\ The code $\mathcal{C}_L(\mathcal{P},aQ)$ is the evaluation of the functions of the  vector space $\mathcal{L}(aQ)$
\begin{equation}\label{code}
\mathcal{C}_L(\mathcal{P},aQ) = \{(\phi(P_1),\ldots,\phi(P_n))\in\ef_q^n\;
\vert\;\phi\in\mathcal{L}(aQ)\},
\end{equation}
and its \textit{dual} is
\begin{equation}\label{codeag}
\mathcal{C}_L(\mathcal{P},aQ)^\perp = \{(c_1,...,c_n)\in\ef_q^n\;\vert\;
\sum_{j=1}^nc_j\phi(P_j) = 0\;\forall \phi\in\mathcal{L}(aQ)\}.\\
\end{equation}
There exists $o_1,\ldots,o_r\in\enn\setminus\{0\}$ such that for a monomial
$M=X_1^{i_1}\cdots X_r^{i_r}$, the \textit{pole order }of $M$ at $Q$ is
\begin{equation*}
v_Q(M) =-(o_1i_1+\cdots +o_ri_r),
\end{equation*}
thus $v_Q(X_i)=-o_i$ for $i=1,\ldots,r$. We  may define the monomial order
\begin{equation*}
  \wdeg(X^\alpha) = \wdeg(X_1^{\alpha_1}\cdots X_r^{\alpha_r}) =(o_1\alpha_1+\cdots+\alpha_ri_r). \end{equation*}
A generating family of $\mathcal{C}_L(\mathcal{P},aQ)$ is then
\begin{equation*}
  \{(X^\alpha(P_1),\ldots,X^\alpha(P_r))\;\vert\;\wdeg(X^\alpha)\leqslant a\},
\end{equation*}
with $X^\alpha(P)=x_1^{\alpha_1}\cdots x_r^{\alpha_r}$, where $P=(1:x_1:\ldots :x_r)$. As a consequence, one has a much simpler form of the code $\mathcal{C}_L(\mathcal{P},aQ)^\perp$:
\begin{equation}\label{agc}
\mathcal{C}_L(\mathcal{P},aQ)^\perp = \{(c_1,...,c_n)\in\ef_q^n\;\vert\;\sum_{i=1}^n c_iX^\alpha(P_i) = 0\;\;\text{for} \;\;\alpha\;\;\text{such that}\;\;\wdeg(\alpha)\leqslant a\}.
\end{equation}
Now, we use the sets $\aaa$ and $\ddd$, defined as in Section \ref{ds-cp}, using the field $\ef$.
\begin{definition}[\cite{HS}]The generalized transform is
\begin{align}\label{gt}
\begin{split}
  \gt : \ef_q^n&\longrightarrow\aaa,\\
  w&\longmapsto
  W(Y) = \sum_{\alpha\in\enn^r}\bigl(\sum_{i=1}^nw_iX^\alpha(P_i)\bigr)Y^\alpha.\\
 \end{split}
 \end{align}
\end{definition}

 This transform  defines an $\ef_q$-injective linear mapping.\\

 Now,  we consider the situation in which a codeword $c$ of our code has been sent through a \textit{communication channel}.  The received word, say $w\in\ef_q^n$ is not necessarily equal to $c$, because of a possible \textit{error} $e$ produced by the channel. We may write
 \begin{equation}\label{err}
  w=c+e.
 \end{equation}
 Of course, the receiver does not know  either $c$ or $e$. The problem is to  find $e$ in order to know $c =w-e$. Instead of finding $e$ directly, one constructs the \textit{syndrome array}.

 \begin{definition}[\cite{HS}]The syndrome array is
 \begin{equation}\label{e-GT}
   E =\gt(e) = E(Y) = \sum_{\alpha\in\enn^r}E_\alpha
   Y^\alpha\in\aaa.
\end{equation}
 \end{definition}

\begin{definition}[\cite{HS}]The \textit{errors locator ideal} is
\begin{equation}
  E^\perp = \{ F(X)\in\ddd\;|\;F(X)\circ E(Y)=0\}\subset\ddd.
\end{equation}
\end{definition}

 We are going to show that if $E\neq 0$, then  $E^\perp\neq \{0\}$, which means that $E$ is a linear recurring sequence (\ref{LRS}). Using \eqref{shift}, this yields
\begin{equation*}
\sum_{\beta\in\enn^r} F_\beta E_{\alpha+\beta} = 0\quad\text{for}\quad \alpha\in\enn^r,
\end{equation*}
where $F(X)=\sum_{\beta\in\enn^r}F_\beta X^\beta$.
  For this purpose, we will need the following lemma:\\
\begin{lemma}[\cite{HS}]\label{ve}
\label{thloc} For an AG code, one has
\begin{equation*}
E^\perp = \mathbf{I}(\supp(e)) = \{\; F(X)\in\ddd\;\vert\;
F(P)=0\;\;\forall\:P\in\supp(e)\;\},
\end{equation*}
where $\supp(e)=\{P_i\in\mathcal{P}\;\vert\; (i\in\{1,\ldots,n\}) \;\ e_i\neq 0\}$.\\
\end{lemma}
We then have
what we need :\\
\begin{corollary}\label{Elrs}If $E\neq 0$, then $E^\perp\neq \{0\}$.\\
\end{corollary}
\begin{proof} If $\supp(e) = \{Q_1,\ldots,Q_m\}\subset\mathcal{P}$
where
\begin{equation*}
Q_j = (a_1^{(j)},\ldots,a_r^{(j)})\in\ef^r\quad\text{for}\quad j=1,\ldots, m,
\end{equation*}
then the polynomial
\begin{equation*}
F(X_1,\ldots,X_r) = \prod_{i=1}^r\prod_{j=1}^m(X_i\;-\;a_i^{(j)})
\end{equation*}
is non-zero and verifies
\begin{equation*}
F(Q_j) =  0\quad\text{for}\quad j=1,\dots m.
\end{equation*}
Thus $F(X)\in \mathbf{I}(\supp(e))$ and by lemma \ref{ve}, it follows that $F(X)\in E^\perp$.\qed
\end{proof}
We have obtained what we need :
\begin{corollary}\label{E-LRS}. The syndrome array $E$ is a linear recurring sequence.

\end{corollary}
\section{Cauchy's equations for the syndrome array}\label{ce}

By Corollary \ref{Elrs}, if $E\neq 0$, the ideal $E^\perp$ is non zero. Let $\leqslant_+$ be the partial order defined on $\enn^r$ by
\begin{equation*}
\alpha =(\alpha,\ldots,\alpha_r)\leqslant_+\beta=(\beta_1,\ldots,\beta_r)\Longleftrightarrow(\forall i\in\{1,\ldots,r\})\;\;\alpha_i\leqslant\beta_i
\end{equation*}
for $\alpha$ and $\beta\in\enn^r$. Then $E^\perp$  has a \textit{Gr\"{o}bner basis}
$\mathcal{G}= \{G_1(X),...,G_k(X)\}$ (with respect to the monomial order $\leqslant_T$ in Section \ref{ds-cp}) where  $G_i(X)\in\ddd$ for
$i=1,\ldots k$ (\cite{Cox,Oberst}). Consider the ``Delta-sets'' (\cite{Faugere,HS,Sakata-0})
\begin{align}\label{Delta-set}
\begin{split}
&\Delta(E^\perp) = \{\alpha\in\enn^r\;\vert\;\;(\exists\; F(X)\in E^\perp),\; \alpha\leqslant_+{\rm LE}(F(X))\},\\
&\Delta(\mathcal{G}) = \{\alpha\in\enn^r\;\vert\; (\exists\; i\in\{1,\ldots,k\} ),\;\alpha\leqslant_+{\rm LE}(G_i(X))\}
\end{split}
\end{align}
and the set
\begin{equation*}
\lep (E^\perp) = \{\lep (F(X))\;\vert\; F(X)\in E^\perp\}.
\end{equation*}
Since $\mathcal{G}$ is a Gr\"{o}bner basis of $E^\perp$, we have
\begin{align*}
\begin{split}
&\Delta(E^\perp) = \Delta(\mathcal{G})\\
&\lep (E^\perp) = \bigcup_{i=1}^n\; (\lep(G_i(X))+\enn^r),\\
\end{split}
\end{align*}
so that
\begin{equation}\label{deltas}
\enn^r = \Delta(E^\perp)\bigcup \lep (E^\perp) =  \Delta(\mathcal{G})\bigcup_{i=1}^n\; (\lep (G_i(X))+\enn^r).
\end{equation}
(\cite{Cox,HS,Oberst}). Let $G(X)$ be the matrix
 \begin{equation*}
  G(X) = \begin{pmatrix}
    G_1(X) \\
    \vdots\\
    G_k(X)\
  \end{pmatrix}\in\ddd^{k,1}
\end{equation*}
and consider the system
\begin{equation}\label{Code-Sys}
S=\{W\in\aaa\;\vert\;G(X)\circ W = 0\}.
\end{equation}
The (unique) column of the matrix $G(X)$ is obviously $\mathbf{K}$-linearly independent, where $\mathbf{K}$ is the field of fractions of $\ddd$. Thus, according to \ref{io}, $S$ is a I/O system, with $p=m=1, Q=0\in\ddd^{k,1}$ and $U=0\in\aaa$.
Therefore, we may, as in \ref{Cauchy-h}, consider the Cauchy's homogeneous equations with respect to $S$.\\

Here is our main theorem:\\
\begin{theorem}\label{main}The syndrome $E$ is the unique solution of the Cauchy's homogeneous equations\\$(G(X),0,\Delta(\mathcal{G}))$:
\begin{equation}\label{cauchy0}
  \begin{cases}
   \quad G(X)\circ E =0, \\
     \quad  E_{\vert\Delta(\mathcal{G})}= V_0,
  \end{cases}
\end{equation}
where $V_0\in\ef_q^{\Delta(\mathcal{G})}$ is an arbitrary element.\\
\end{theorem}
\begin{proof} We are going to prove that \eqref{cauchy0} is verified by  all element $W$ of $S$, hence true for the particular case $W=E$.  The first equation of \eqref{cauchy0} follows from the construction of $S$. Now,  write $\Gamma=\Delta(\mathcal{G})$.
 Each trajectory $W$ of $S$ is then uniquely determined by its restriction to $\Gamma$, which is $V_0$. Indeed, suppose that  $W_\alpha$ is known and is equal to $V_{0\alpha}$ for
$\alpha\in\Gamma$. We are going to calculate  $W_\alpha$ by \textit{n\oe therian {\rm or} transfinite induction} (see \cite{Oberst}) on  $\alpha\in\enn^r\setminus\Gamma$. Let
$\alpha_0=\min_{\leqslant_T}(\enn^r\setminus\Gamma)$. Using \eqref{deltas}, there exists $G_k(X)\in\mathcal{G}$ such that  $\alpha_0$ is an entry with respect to  $G_k(X)$, i.e. there exists $t\in\enn^r$ such that $\alpha_0=t+d$ with $d=\lep G_k(X)$. Since  $G_k(X)\circ W=0$, we then have
\begin{equation}\label{equat1}
  \sum_{\alpha\leqslant_T d} G_{k\alpha}W_{\alpha+t} = 0,
  \end{equation}
and
\begin{equation}\label{equat2}
W_{\alpha_0} = W_{t+d} = -\frac{1}{G_{kd_k}}\sum_{\alpha<_T d}
G_{k\alpha}W_{\alpha+t}.
\end{equation}
But, since
\begin{equation*}
 \alpha<_T d\Longrightarrow\alpha+t<_Tt+ d = \alpha_0,
\end{equation*}
and by the choice of $\alpha_0$, we necessarily have
$\alpha+t\in\Gamma$. Thus, $W_{\alpha+t}=V_{\alpha+t}$ is already known and $W_{\alpha_0}$  can be calculated by {\eqref{equat2}}
for $\alpha_0=\min_{\leqslant_T}(\enn^r\setminus\Gamma)$. Now, let  $\alpha\in\enn^r\setminus\Gamma$ and suppose, by the recurrence hypothesis that $W_v$ is already calculated for  $v$ with $\alpha_0\leqslant_T v< \alpha$.  Using again \eqref{deltas} there exists $t\in\enn^r$
and $G_l(X)\in\mathcal{G}$ such that $\alpha=t+d$, with $d=\lep
(G_l)$. As in {\eqref{equat2}}, we have
\begin{equation}\label{equat3}
W_\alpha = W_{t+d} = -\frac{1}{G_{kd}}\sum_{\alpha<_T d}
G_{k\alpha}W_{\alpha+t},
\end{equation}
and $W_{\alpha+t}$ is already known by the recurrence hypothesis, since we have $\alpha+t<_T\alpha+d=\alpha$. Thus $W_\alpha$ can be  uniquely calculated by\eqref{equat3}. Therefore, by n{\oe}therian recurrence, we can calculate $W_\alpha$ for $\alpha\in\enn^r\setminus\Gamma$.
\end{proof}

Now, consider the one dimensional case  $r=1$. Let :\\
  $\bullet$ $E$ be the generalized transform of the error
  $e$,\\
  $\bullet$ $F(X)$ be the \textit{characteristic polynomial} of $E^\perp$  and $d=\deg
  F(X)>1$,\\
 $\bullet$ $S=\kr F(X)$.\\
 Then $\Gamma$ and $\Delta(E^\perp)$ are of the following forms\\
 \begin{equation*}
 \Gamma = \{0,\ldots,d-1\}\subset\enn,\quad \Delta(E^\perp) = \Delta(F(X)) =\{V_0,\ldots,V_{d-1}\}
 \end{equation*}
and we have a simpler version of lemma \ref{cauchy0}: \\

 \textit{Every element $W\in S$ is the unique solution of the Cauchy's equations}
\begin{equation*}
 \begin{cases}
     \quad F(X)\circ W =0, \\
     \quad  W_{\vert\{0,\ldots,d-1\}}= (V_0,\ldots,V_{d-1}\}\in\ef_q^d,
  \end{cases}
\end{equation*}
We can directly calculate $W$ with $F(X)$ and
$V$. Indeed, write $F(X)=\sum_{i=1}^d F_iX^i$ with $d=\deg
F(X)$ et $F_d=1$. We have $W_h=V_h$ for $h\leqslant d-1$. For $k\in\enn$, we have :
\begin{equation*}\
  \sum_{i=1}^d F_iW_{i+k} = 0\quad\text{et}\quad W_n = W_{d+k} = -\sum_{1\leqslant i<d}F_iW_{i+k},
\end{equation*}
and this defines $W_n$ using $W_h$, with $h<n$.\qed

We may consider \eqref{cauchy0} as the fundamental equation which lies behind the BMS algorithm in the decoding process. However, at the beginning, the matrix $G(X)$ in \eqref{cauchy0},  is of course unknown, because it is constructed from the unknown syndrome $E$.  But, by \eqref{err}, we have $\gt(w)=\gt(c+e)= \gt(c)+\gt(e)$. Using \eqref{agc} and \eqref{gt}, we have $[\gt(c)]_\alpha = 0$ whenever $\wdeg(\alpha)\leqslant a$  (where $[W]_\alpha$ also denotes the coefficient of the power series $W\in\ef_q[Y]$ with respect to $Y^\alpha$). Let $Z$ be the set
\begin{equation}\label{zero}
  Z = \{\alpha\in\enn^r\;\vert\; \wdeg(\alpha)\leqslant a\}.
\end{equation}
We then have $[\gt(w)]_\alpha=[\gt(e)]_\alpha = E_\alpha$ for $\alpha\in Z$, so that $E_\alpha$ is known  on the set $Z$ only since it is equal to $[\gt(w)]_\alpha$ and $w$ is known.\\

The general idea of the BMS algorithm is  to  use these known terms of $E$ to construct some polynomials, which are valid \textit{recurrence relations} for theses terms. Then, using these polynomials, the algorithm calculates more terms of $E$ and so on. Finally, the algorithm finds a Gr\"{o}bner basis of the ideal $E^\perp$, which, in turn, by \eqref{cauchy0}, allows to calculate $E$, and $e$, using the inverse of the $\gt$ transform.


\end{document}